\documentclass[12pt,fleqn]{article}

\usepackage{amsmath,amssymb,amsthm}
\usepackage[margin=1in]{geometry}
\usepackage{xcolor}
\usepackage{url}
\usepackage{float}
\usepackage{array}
\usepackage{comment}
\usepackage{booktabs}
\usepackage[hypertexnames=false,bookmarksnumbered=true,final]{hyperref}
\usepackage[capitalize,sort]{cleveref}

\usepackage{amsmath}

\newcommand*{\thmdep}[2]{}

\newcommand*{\wLoG}{without loss of generality}

\let\eps\varepsilon

\newcommand*{\chat}{\widehat{c}}
\newcommand*{\ghat}{\widehat{g}}
\newcommand*{\phat}{\widehat{p}}

\newcommand*{\Ical}{\mathcal{I}}

\newcommand*{\alphahat}{\widehat{\alpha}}

\newcommand*{\defeq}{:=}

\DeclareMathOperator*{\E}{E}

\DeclareMathOperator*{\argmin}{argmin}
\DeclareMathOperator*{\argmax}{argmax}

\usepackage{iftex}
\usepackage{url}
\usepackage{mathtools}
\usepackage{algorithm,algpseudocode}
\usepackage{comment}

\def\colorschemesepia{sepia}
\def\colorschemedark{dark}
\def\colorschemelight{light}

\ifx\colorscheme\undefined
\let\colorscheme\colorschemelight
\fi

\ifx\colorscheme\colorschemelight
\colorlet{textColor}{black}
\colorlet{bgColor}{white}
\fi

\ifx\colorscheme\colorschemesepia
\definecolor{textColor}{HTML}{433423}
\definecolor{bgColor}{HTML}{fbf0da}
\fi

\ifx\colorscheme\colorschemedark
\definecolor{textColor}{HTML}{bdc1c6}
\definecolor{bgColor}{HTML}{202124}
\definecolor{textBlue}{HTML}{8ab4f8}
\definecolor{textRed}{HTML}{f9968b}
\definecolor{textGreen}{HTML}{81e681}
\definecolor{textPurple}{HTML}{c58af9}
\else
\colorlet{textBlue}{blue!50!black}
\colorlet{textRed}{red!50!black}
\colorlet{textGreen}{green!50!black}
\definecolor{textPurple}{HTML}{681da8}
\fi

\ifx\colorscheme\colorschemelight\else
\pagecolor{bgColor}
\color{textColor}
\fi

\hypersetup{
colorlinks,
linkcolor=textRed,
citecolor=textRed,
urlcolor=textBlue
}

\makeatletter
\g@addto@macro{\UrlBreaks}{%
\do\/%
\do\a\do\b\do\c\do\d\do\e\do\f\do\g\do\h\do\i\do\j\do\k\do\l\do\m%
\do\n\do\o\do\p\do\q\do\r\do\s\do\t\do\u\do\v\do\w\do\x\do\y\do\z%
\do\A\do\B\do\C\do\D\do\E\do\F\do\G\do\H\do\I\do\J\do\K\do\L\do\M%
\do\N\do\O\do\P\do\Q\do\R\do\S\do\T\do\U\do\V\do\W\do\X\do\Y\do\Z%
\do\0\do\1\do\2\do\3\do\4\do\5\do\6\do\7\do\8\do\9%
}
\makeatother

\allowdisplaybreaks
\newcolumntype{L}{>{$\displaystyle}l<{$}}
\algnewcommand{\LineComment}[1]{\State \textcolor{gray}{\texttt{//} \textit{#1}}}
\algdef{SE}[REPEATN]{RepeatN}{End}[1]{\algorithmicrepeat\ #1 \textbf{times}}{\algorithmicend}

\makeatletter
\@ifpackageloaded{enumitem}{%
\newenvironment*{tightemize}{\begin{itemize}[noitemsep]}{\end{itemize}}%
\newenvironment*{tightenum}{\begin{enumerate}[noitemsep]}{\end{enumerate}}%
}{%
}
\makeatother

\newtheorem{theorem}{Theorem}
\newtheorem{definition}{Definition}

\newtheorem{corollary}[theorem]{Corollary}
\newtheorem{lemma}[theorem]{Lemma}

\newtheorem{observation}[theorem]{Observation}
\newtheorem{example}[definition]{Example}

\crefname{claim}{Claim}{Claims}
\crefname{property}{Property}{Properties}
\crefname{example}{Example}{Examples}
\crefname{observation}{Observation}{Observations}
\crefname{transformation}{Transformation}{Transformations}

\title{EF1 for Mixed Manna with Unequal Entitlements}
\author{
Jugal Garg%
\thanks{Department of Industrial \& Enterprise Engineering, University of Illinois at Urbana-Champaign, USA}
\\ \texttt{\small jugal@illinois.edu}
\and
Eklavya Sharma\footnotemark[1]
\\ \texttt{\small eklavya2@illinois.edu}
}
\date{\empty}

\begin{document}

\maketitle

\begin{abstract}
We study fair division of indivisible mixed manna when agents have unequal entitlements,
with weighted envy-freeness up to one item (WEF1) as our primary notion of fairness.
We identify several shortcomings of existing techniques to achieve WEF1.
Hence, we relax WEF1 to weighted envy-freeness up to 1 transfer (WEF1T),
and give a polynomial-time algorithm for achieving it.
We also generalize Fisher markets to the mixed manna setting,
and use them to get a polynomial-time algorithm for two agents that outputs a WEF1 allocation.

\end{abstract}

\section{Introduction}
\label{sec:intro}

Fair division of indivisible items is a central problem in economics and game theory.
Here $m$ items must be divided among $n$ agents in a \emph{fair} way,
i.e., no agent is unreasonably favored over another agent.

Most existing work on fair division can be divided into two broad settings:
fair division of goods, and fair division of chores.
However, there is another interesting setting, called \emph{mixed manna}.
Here the set of items contains both goods and chores.
Moreover, whether an item is a good or a chore is subjective across the agents.
E.g., when distributing teaching responsibilities among faculty members in a university,
someone who likes teaching that course would see it as a good,
and someone who dislikes teaching that course would see it as a chore.

Another twist to classic fair division is when the agents have \emph{unequal entitlements}
(also called the \emph{weighted} or \emph{asymmetric} setting).
E.g., if multiple people invest different amounts into a joint venture,
it is reasonable to demand that returns be divided among them
proportional to the investment amount.
During inheritance division, closer relatives expect to receive a larger share of the inheritance.

Although fair division of mixed manna has been studied before,
and weighted fair division of either goods or chores has been studied before,
not much investigation has gone into weighted fair division of mixed manna.
Hence, in this work, we study fair division of mixed manna
when agents have unequal entitlements and additive valuations.

Formally, in fair division, there are $n$ agents and a set $M$ of items.
Each agent $i$ has a \emph{valuation function} $v_i: 2^M \to \mathbb{R}$ where
$v_i(S)$ is a number denoting how much agent $i$ likes the set $S \subseteq M$ of items.
Each agent $i$ also has a positive real number $w_i$ denoting her entitlement.
Our goal is to output an allocation $A$, i.e., a partition $(A_1, \ldots, A_n)$ of the items, that is \emph{fair}.

The notion of fairness we consider is weighted envy-freeness up to one item (WEF1).
An allocation $A$ is EF1 if for every ordered pair $(i, j)$ of agents,
\begin{align*}
& \frac{v_i(A_i)}{w_i} \ge \frac{v_i(A_j)}{w_j},
\quad\textrm{or}\; \frac{v_i(A_i)}{w_i} \ge \frac{v_i(A_j \setminus \{g\})}{w_j} \textrm{ for some } g \in A_j,
\\ &\qquad\textrm{or}\; \frac{v_i(A_i \setminus \{c\})}{w_i} \ge \frac{v_i(A_j)}{w_j} \textrm{ for some } c \in A_i.
\end{align*}

In addition to fairness, we also study Pareto optimality (PO).
An allocation is pareto optimal if agents cannot rearrange items among themselves in a way
that no one becomes worse off and someone becomes better off.

\subsection{Our Results}
\label{sec:our-results}

In \cref{sec:ef1}, we first give a polynomial-time algorithm whose output is
weighted envy-free up to 1 transfer (WEF1T), which is a weaker notion of fairness than WEF1.
Although we couldn't prove or disprove the existence of WEF1 allocations for mixed manna,
we present insights into the limitations of existing techniques.

Several results for EF1+PO allocations for goods and for chores are based on Fisher market equilibria.
In \cref{sec:ef1-po}, we generalize Fisher market equilibria to our setting of weighted mixed manna.
Interestingly, market equilibria in this setting can behave in ways
slightly different from what one would expect in the goods-only and chores-only settings.
Then, we give a polynomial time algorithm based on market equilibrium
that outputs a WEF1+PO allocation for the special case of two agents.

\subsection{Related Work}

Fair division using EF1 as the notion of fairness has been the subject of a lot of research,
and a wide variety of techniques have been used to attack the problem and its special cases.
See \cref{table:known-results} for a summary of the results.

\begin{table}[htb]
\centering
\caption{Prior work on WEF1 and WEF1+PO.
Gray text indicates work subsumed by subsequent results in the table.}
\vspace{0.5em}
\label{table:known-results}
\newcommand*{\obsolText}{\color{textColor!30!bgColor}}
\begin{tabular}{ccccccc}
\toprule
fairness notion &  & polytime & constraints & vfunc & technique\textsuperscript{\textdagger}
\\ \midrule
   \obsolText EF1 & \obsolText goods & \obsolText yes & \obsolText -- & \obsolText additive & \obsolText rr (folk)
\\ EF1 & goods & yes & -- & monotone & ece \cite{lipton2004approximately}
\\ \obsolText EF1+PO & \obsolText goods & \obsolText no & \obsolText -- & \obsolText additive & \obsolText mnw \cite{caragiannis2019unreasonable}
\\ WEF1 & goods & yes & -- & additive & pseq \cite{chakraborty2021weighted}
\\ WEFX & goods & yes & -- & identical & ece \cite{springer2024almost}
\\ WEF1+PO & goods & pseudo & -- & additive & meq \cite{chakraborty2021weighted}
\\ WEF1+PO & goods & yes & $n=2$ & additive & aw \cite{chakraborty2021weighted}
\\ WWEF1+PO & goods & no & -- & additive & mnw \cite{chakraborty2021weighted}
\\ \midrule
   \obsolText EF1 & \obsolText chores & \obsolText yes & \obsolText -- & \obsolText additive & \obsolText rr (folk)
\\ \obsolText EF1+PO & \obsolText chores & \obsolText yes & \obsolText $n=3$ & \obsolText additive & \obsolText meq \cite{garg2023new}
\\ \obsolText EF1+PO & \obsolText chores & \obsolText yes & \obsolText 2 util funcs & \obsolText additive & \obsolText meq \cite{garg2023new}
\\ EFX+PO & chores & yes & $n=3$ & bival add & meq \cite{garg2023new}
\\ WEF1 & chores & yes & -- & additive & pseq \cite{springer2024almost}
\\ WEF1+PO & chores & yes & 3 util funcs & additive & meq \cite{garg2024weighted}
\\ WEF1+PO & chores & yes & 2 chore types & additive & meq \cite{garg2024weighted}
\\ \midrule
   WPROP1+PO & mixed & yes & -- & additive & meq \cite{aziz2020polynomial}
\\ EF1 & mixed & yes & -- & additive & pseq \cite{aziz2021fair}
\\ EF1 & mixed & yes & -- & dbl mono & ece \cite{bhaskar2021approximate}
\\ EF1+PO & mixed & yes & $n=2$ & additive & aw \cite{aziz2021fair}
\\ \bottomrule
\end{tabular}

\textsuperscript{\textdagger}\footnotesize folk: unpublished but well-known,
rr: round robin,
ece: envy cycle elimination,
pseq: picking sequence,
meq: Fisher market equilibrium,
aw: adjusted winner,
mnw: maximum Nash welfare.
\end{table}

Several other fairness notions are related to EF1.
EFX (envy-free up to any good) is a strengthening of EF1.
For goods, an allocation is EFX if between every pair $(i, j)$ of agents
and every good $g$ in $j$'s bundle, $i$ stops envying $j$ if $g$ is removed from $j$'s bundle.
For chores, an allocation is EFX if between every pair $(i, j)$ of agents
and every chore $c$ in $i$'s bundle, $i$ stops envying $j$ if $c$ is removed from $i$'s bundle.
The existence of EFX allocations (for both the goods setting and the chores setting)
is a major open problem in fair division.
\cite{springer2024almost} studied EFX in the weighted setting (called WEFX)
and showed that WEFX allocations are not guaranteed to exist.

Another related notion of fairness is \emph{PROP1}.
An allocation is PROP1 if for every agent $i$, her bundle's value is at least her proportional share
minus the value of some good not in her bundle or minus the disutility of some chore in her bundle.
Agent $i$'s proportional share is her value for the entire set of items divided by $n$.
For unequal entitlements, we use the term WPROP1, and agent $i$'s proportional share
is her value for the entire set of items multiplied by $w_i / (\sum_{j=1}^n w_j)$,
where $w_j$ is agent $j$'s entitlement.
\cite{aziz2021fair} showed that an EF1 allocation of mixed manna is also PROP1,
but \cite{chakraborty2021weighted} showed that a WEF1 allocation of goods may not be WPROP1.

\cite{chakraborty2021weighted} introduced a notion called WWEF1 (weakly WEF1),
which is a relaxation of WEF1, but is the same as EF1 for equal entitlements.
They showed that an allocation maximizing the weighted Nash welfare is WWEF1,
which generalizes the result of \cite{caragiannis2019unreasonable} for the unweighted setting.

Although market equilibria have been studied for mixed manna before
\cite{bogomolnaia2017competitive,chaudhury2021competitive}, we define them slightly differently.
We don't include budgets in our formultion, and our definition is easier to use
when agents have additive valuations.

\section{Preliminaries}

For any $t \in \mathbb{Z}_{\ge 0}$, define $[t] \defeq \{1, 2, \ldots, t\}$.

\subsection{Fair Division Instances}

A fair division instance is given by the tuple $(N, M, V, w)$,
where $N$ is the set of agents, $M$ is the set of items,
$V = (v_i)_{i \in N}$ is the \emph{valuation profile},
i.e., $v_i: 2^M \to \mathbb{R}$ is agent $i$'s valuation function,
and $w = (w_i)_{i \in N}$ is the entitlement vector,
where $w_i \in \mathbb{R}_{>0}$ is agent $i$'s entitlement.
Often, we assume \wLoG{} that $N = [n]$ and $M = [m]$.

For notational convenience, for any item $t$, we denote $v_i(\{t\})$ by $v_i(t)$.
A valuation function $v: 2^M \to \mathbb{R}$ is \emph{additive}
if for any set $S$ of items, we have $v(S) = \sum_{j \in S} v(j)$.
Unless specified otherwise, we assume that all valuation functions are additive.

We can classify any item $j \in M$ as a good, neutral, or a chore:
\begin{enumerate}
\item $j$ is a \emph{good} if $v_i(j) > 0$ for some $i \in N$.
\item $j$ is \emph{neutral} if $v_i(j) = 0$ for some $i \in N$ and $v_i(j) \le 0$ for all $i \in N$.
\item $j$ is a \emph{chore} if $v_i(j) < 0$ for all $i \in N$.
\end{enumerate}
Additionally, a good $j$ is called \emph{pure} if $v_i(j) \ge 0$ for all $i \in N$.

For a fair division instance $\Ical \defeq ([n], [m], (v_i)_{i=1}^n, w)$,
a fractional allocation $x$ is one where for each item $j$,
agent $i$ gets an $x_{i,j}$ fraction of item $j$.
Hence, $\sum_{i=1}^n x_{i,j} = 1$.
For a fractional allocation $x$, the vector $x_i \defeq (x_{i,1}, \ldots, x_{i,m})$
is called agent $i$'s \emph{bundle} in $x$.
For any vector $z \in \mathbb{R}^m_{\ge 0}$ and valuation function $v$,
let $v(z) \defeq \sum_{j=1}^m v(j)z_j$.
Hence, agent $i$'s value for agent $j$'s bundle is denoted as $v_i(x_j)$.

If $x_{i,j} \in \{0, 1\}$ for all $i \in [n]$ and $j \in [m]$,
then $x$ is said to be \emph{integral}, and can alternatively be expressed as
a tuple $A \defeq (A_1, \ldots, A_n)$, where $A_i \defeq \{j \in [m]: x_{i,j} = 1\}$
is the set of items allocated to agent $i$, called agent $i$'s \emph{bundle}.
Unless specified otherwise, all allocations are assumed to be integral.

\subsection{Fairness}

\begin{definition}
Let $A$ be an allocation for a fair division instance $(N, M, (v_i)_{i=1}^n, w)$.
\begin{enumerate}
\item Agent $i$ envies agent $j$ in $A$ if
    \[ \frac{v_i(A_i)}{w_i} < \frac{v_i(A_j)}{w_j}. \]
    Allocation $A$ is EF (envy-free) if no agent envies any other agent.

\item Agent $i$ EF1-envies agent $j$ in $A$ if $i$ envies $j$ in $A$ and
    \[ \frac{v_i(A_i \setminus \{t\})}{w_i} < \frac{v_i(A_j)}{w_j} \;\forall t \in A_i,
    \qquad\textrm{and}\qquad
    \frac{v_i(A_i)}{w_i} < \frac{v_i(A_j \setminus \{t\})}{w_j} \;\forall t \in A_j. \]
    Allocation $A$ is EF1 (envy-free up to 1 item) if no agent EF1-envies any other agent.

\item Agent $i$ EF1T-envies agent $j$ if $i$ envies $j$ in $A$ and
    \[ \frac{v_i(A_i \setminus \{t\})}{w_i} < \frac{v_i(A_j \cup \{t\})}{w_j} \;\forall t \in A_i,
    \qquad\textrm{and}\qquad
    \frac{v_i(A_i \setminus \{t\})}{w_i} < \frac{v_i(A_j \setminus \{t\})}{w_j} \;\forall t \in A_j. \]
    Allocation $A$ is EF1T (envy-free up to 1 transfer) if no agent EF1T-envies any other agent.
\end{enumerate}
\end{definition}

We use the terms WEF, WEF1, and WEF1T when we consider the setting with unequal entitlements
and use the terms EF, EF1, and EF1T when we consider the setting with equal entitlements.

\subsection{Efficiency}

\begin{definition}[PO and fPO]
A fractional allocation $x$ is said to \emph{Pareto-dominate}
another fractional allocation $y$ in instance $\Ical \defeq ([n], [m], (v_i)_{i=1}^n, w)$
iff $v_i(x_i) \ge v_i(y_i)$ for all $i \in [n]$
and $v_i(x_i) > v_i(y_i)$ for some $i \in [n]$.

An integral allocation is \emph{Pareto-optimal} (PO) in $\Ical$ iff
it is not Pareto-dominated in $\Ical$ by any other integral allocation.
A (fractional or integral) allocation is \emph{fractionally Pareto-optimal} (fPO) in $\Ical$
iff it is not Pareto-dominated in $\Ical$ by any other fractional allocation.
\end{definition}

Note that any integral fPO allocation is also PO.

\section{WEF1 and WEF1T}
\label{sec:ef1}

We first show how to efficiently compute WEF1T allocations for mixed manna.
Then we turn to the problem of computing WEF1 allocations of mixed manna
and illustrate the limitations of existing techniques.

\subsection{Algorithm for WEF1T}
\label{sec:ef1:sep}

We show that if we compute a WEF1 allocation of the goods
and a WEF1 allocation of the chores independently,
then the resulting allocation of the mixed manna instance
is weighted envy-free up to one transfer (WEF1T).
Since WEF1 allocations for goods and chores can be computed in polynomial time
\cite{chakraborty2021weighted,aziz2021fair}, we can thus compute WEF1T allocations
for mixed manna in polynomial time too.

\begin{theorem}
\label{thm:wef1t}
Let $([n], G \cup C, (v_i)_{i=1}^n, w)$ be a fair division instance,
where $G$ is the set of goods and neutral items, and $C$ is the set of chores.
Let $A^{(G)}$ be a WEF1 allocation of $G$ and $A^{(C)}$ be a WEF1 allocation of $C$.
Then allocation $A$ is WEF1T, where $A_i \defeq A^{(G)}_i \cup A^{(C)}_i$ for every agent $i \in [n]$.
\end{theorem}
\begin{proof}
Let $i$ and $j$ be any two agents. Since $A^{(G)}$ and $A^{(C)}$ are WEF1,
we get that for some $c \in A^{(C)}_i$ and $g \in A^{(G)}_j$, we have
\begin{align*}
& \frac{v_i(A_i \setminus \{c\})}{w_i}
= \frac{v_i(A^{(G)}_i)}{w_i} + \frac{v_i(A^{(C)}_i \setminus \{c\})}{w_i}
\ge \frac{v_i(A^{(G)}_j \setminus \{g\})}{w_j} + \frac{v_i(A^{(C)}_j)}{w_j}
= \frac{v_i(A_j \setminus \{g\})}{w_j}.
\end{align*}
If $d_i(c) \ge v_i(g)$, then
$v_i(A_i \setminus \{c\})/w_i \ge v_i(A_j \setminus \{g\})/w_j \ge v_i(A_j \cup \{c\})/w_j$.
If $d_i(c) \le v_i(g)$, then
$v_i(A_i \cup \{g\})/w_i \ge v_i(A_i \setminus \{c\})/w_i \ge v_i(A_j \setminus \{g\})/w_j$.
Hence, $A$ is WEF1T.
\end{proof}

One technicality of our setting is that some goods may be non-pure.
We can easily modify \cite{chakraborty2021weighted} to handle this and get a WEF1 allocation for goods:
have each agent $i$ skip her turn in the picking sequence if $v_i(g) \le 0$
for every unallocated good $g$ (or, equivalently, create a dummy item of value 0
and give it to agent $i$).

\subsection{Impossibility for Ordinal WEF1}
\label{sec:ef1:hard-sd}

Many algorithms for (W)EF1 allocation are based on picking sequences
\cite{chakraborty2021weighted,springer2024almost,aziz2021fair},
i.e., agents are repeatedly asked to pick their favorite item from a subset of items.
Hence, the output of these algorithms only depend on each agent's preference order over the items
(and not on the precise valuations of the items).
We show that such algorithms cannot help us get WEF1 for mixed manna.
To formalize this result, we use the concept of \emph{ordinal compatibility} of instances.

\begin{definition}
Two additive valuation functions $u, v: 2^M \to \mathbb{R}$ are \emph{ordinally compatible}
iff $\forall t_1, t_2 \in M$, we have $u(t_1) > u(t_2) \iff v(t_1) > v(t_2)$,
and for all $t \in M$, we have $u(t) > 0 \iff v(t) > 0$ and $u(t) < 0 \iff v(t) < 0$.

Two fair division instances $(N, M, V, w)$ and $(N, M, V', w')$ are ordinally compatible iff
$v_i$ and $v'_i$ are ordinally compatible for all $i \in N$.
\end{definition}

\begin{theorem}
\label{thm:ord-compat-3g1c}
There are four ordinally compatible fair division instances,
each with 2 agents, 3 goods, one chore, and identical valuations,
such that no allocation is WEF1 for all of them.
\end{theorem}
\begin{proof}
Let $G \defeq \{g_1, g_2, g_3\}$ be the set of goods,
$C \defeq \{c\}$ be the set of chores, and $M \defeq G \cup C$.
Let $0 < \eps < 1/4$. For $t \in [4]$, define the instance
$\Ical_t \defeq ([2], M, (v^{(t)}, v^{(t)}), w)$,
where $v^{(t)}(g_1) = 1 + 2\eps$, $v^{(t)}(g_2) = 1 + \eps$,
$v^{(t)}(g_3) = \begin{cases}\eps & \textrm{ if } t \textrm{ is odd}
\\ 1 & \textrm{ if } t \textrm{ is even}\end{cases}$,
and $v^{(t)}(c) = \begin{cases}-\eps & \textrm{ if } t \le 2
\\ -3 & \textrm{ if } t \ge 3\end{cases}$.
See \cref{table:ord-compat-3g1c}.
Assume some allocation $A$ is WEF1 for every $\Ical_t$ for $t \in [4]$.

\begin{table}[htb]
\centering
\caption{Four ordinally compatible valuation functions. Here $0 < \eps < 1/4$.}
\label{table:ord-compat-3g1c}
\begin{tabular}{ccccc}
\toprule $t$ & $v^{(t)}(g_1)$ & $v^{(t)}(g_2)$ & $v^{(t)}(g_3)$ & $v^{(t)}(c)$
\\ \midrule $1$ & $1+2\eps$ & $1+\eps$ & $\eps$ & $-\eps$
\\ $2$ & $1+2\eps$ & $1+\eps$ & $1$ & $-\eps$
\\ $3$ & $1+2\eps$ & $1+\eps$ & $\eps$ & $-3$
\\ $4$ & $1+2\eps$ & $1+\eps$ & $1$ & $-3$
\\ \bottomrule
\end{tabular}
\end{table}

Then $g_1$ and $g_2$ cannot be allocated to the same agent,
otherwise the other agent will be EF1-envious in $\Ical_1$.
Agent 1 cannot have 2 goods, otherwise agent 2 will be EF1-envious in $\Ical_2$.
Hence, $|A_1 \cap G| = 1$, $|A_2 \cap G| = 2$, and $g_3 \in A_2$.

If $c \in A_2$, then agent 2 EF1-envies agent 1 in $\Ical_3$.
If $c \in A_1$, then agent 1 EF1-envies agent 2 in $\Ical_4$.
This no allocation $A$ can be WEF1 for every $\Ical_t$ for $t \in [4]$.
\end{proof}

\subsection{Hard Examples for Two-Phase Algorithms}
\label{sec:ef1:hard-ext}

\cite{aziz2021fair} computes an EF1 allocation of mixed manna by first
computing an EF1 allocation of chores, and then extending it to goods.
We show that one cannot extend a WEF1 allocation of chores
to a WEF1 allocation of the mixed manna instance.

\begin{example}
Let $N = [2]$, $w = (2, 3)$, $M \defeq \{c, g_1, g_2\}$,
$d_1(c) = 1-\eps$, $d_2(c) = 1+\eps$, and $v_i(g_1) = v_i(g_2) = 1$ for $i \in N$.

Suppose we give $c$ to agent 2.
If agent 2 receives both $g_1$ and $g_2$, then agent 1 will EF1-envy her.
If agent 2 receives at most one of $g_1$ and $g_2$, then she will EF1-envy agent 1.
Hence, we can't begin by giving $c$ to agent 2.
\end{example}

Similarly, we show that one cannot extend a WEF1 allocation of goods
to an allocation of the mixed manna instance.

\begin{example}
Let $N = [2]$, $w = (2, 3)$, $M \defeq \{g, c_1, c_2\}$,
$v_1(g) = 1+\eps$, $v_2(g) = 1-\eps$, and $d_i(c_1) = d_i(c_2) = 1$ for $i \in N$.

Suppose we give $g$ to agent 2.
If agent 2 receives both $c_1$ and $c_2$, then she will EF1-envy agent 1.
If agent 2 receives at most one of $c_1$ and $c_2$, then agent 1 will EF1-envy her.
Hence, we can't begin by giving $g$ to agent 2.
\end{example}

\section{WEF1+PO Alloction Using Fisher Market}
\label{sec:ef1-po}

We extend the concepts of Fisher markets and market equilibria to the mixed manna setting.
Then we give a polynomial time algorithm based on market equilibria
to compute a WEF1+fPO allocation for mixed manna when there are only two agents.

\subsection{Fisher Market}

Let $\Ical \defeq ([n], [m], (v_i)_{i=1}^n, w)$ be a fair division instance.
A \emph{Fisher market} for $\Ical$ is a pair $(x, p)$,
where $x$ is a fractional allocation for $\Ical$
and $p \in \mathbb{R}^m$ is called the \emph{price vector}.
$(x, p)$ is called \emph{integral} if $x$ is an integral allocation.
For any vector $z \in \mathbb{R}^m_{\ge 0}$, let $p(z) \defeq \sum_{j=1}^m p_jz_j$.
Then $p(x_i)$ is called agent $i$'s \emph{budget}.
For a set $S$ of items, define $p(S) \defeq \sum_{j \in S} p_j$.

\begin{definition}[market equilibrium]
A Fisher market $(x, p)$ is called a market equilibrium if both of the following hold:
\begin{enumerate}
\item For every item $j$, $p_j > 0$ if $j$ is a good, $p_j < 0$ if $j$ is a chore,
    and $p_j = 0$ if $j$ is neutral.
\item For every agent $i \in [n]$, there exists $\alpha_i \in \mathbb{R}_{>0}$ such that
    $v_i(j) \le \alpha_i p_j$ for all $j \in [m]$,
    and $v_i(j) = \alpha_i p_j$ when $x_{i,j} > 0$.
    ($\alpha_i$ is called agent $i$'s \emph{best-bang-per-buck} (BBB).)
\end{enumerate}
\end{definition}

Note that when an agent $i$'s bundle has only neutral items,
her best-bang-per-buck $\alpha_i$ may not be unique.
This does not happen in the goods-only and chores-only setting.

Market equilibria are useful because they help us get fPO allocations.
We prove this in the following result.

\begin{lemma}
\label{thm:first-welfare}
If $(x, p)$ is a market equilibrium, then $x$ is fPO.
\end{lemma}
\begin{proof}
For any agent $i$, we have $v_i(x_i) = \alpha_i p(x_i)$.
For any vector $z \in \mathbb{R}_{\ge 0}^m$, we have $v_i(z) \le \alpha_i p(z)$.
Assume $x$ is not fPO.
Suppose a fractional allocation $y$ Pareto-dominates $x$.

\textbf{Case 1}: $p(y_i) < p(x_i)$ for some $i \in [n]$.
\\ Then $v_i(y_i) \le \alpha_i p(y_i) < \alpha_i p(x_i) = v_i(x_i)$.
Hence, $y$ does not Pareto-dominate $x$.

\textbf{Case 2}: $p(y_i) = p(x_i)$ for all $i \in [n]$.
\\ Then $v_i(y_i) \le \alpha_i p(y_i) = \alpha_i p(x_i) = v_i(x_i)$ for all $i$.
Hence, $y$ does not Pareto-dominate $x$.

Hence, we get a contradiction, so $x$ is fPO.
\end{proof}

\begin{definition}[pWEF1]
Let $(A, p)$ be an integral market equilibrium
for instance $\Ical \defeq ([n], [m], (v_i)_{i=1}^n, w)$.
Agent $i$ pWEF1-envies agent $j$ if all of the following hold:
\begin{enumerate}
\item $\displaystyle \frac{p(A_i)}{w_i} < \frac{p(A_j)}{w_j}$.
\item $\forall g \in A_j$,
    $\displaystyle \frac{p(A_i)}{w_i} < \frac{p(A_j \setminus \{g\})}{w_j}$.
\item $\forall c \in A_i$,
    $\displaystyle \frac{p(A_i \setminus \{c\})}{w_i} < \frac{p(A_j)}{w_j}$.
\end{enumerate}
$(A, p)$ is pWEF1 (price WEF1) if no agent pWEF1-envies any other agent.
\end{definition}

The following results show that to get a WEF1+fPO allocation,
it suffices to find a pWEF1 market equilibrium.

\begin{lemma}
\label{thm:ef1-envy-implies-pef1-envy}
Let $(A, p)$ be a market equilibrium for instance $\Ical \defeq ([n], [m], (v_i)_{i=1}^n, w)$.
If agent $i$ WEF1-envies agent $j$ in $A$, then $i$ pWEF1-envies $j$ in $(A, p)$.
\end{lemma}
\begin{proof}
Let $\alpha_i \in \mathbb{R}_{>0}$ be agent $i$'s best-bang-per-buck in $(A, p)$. Then
\begin{enumerate}
\item $\displaystyle \frac{\alpha_i p(A_i)}{w_i} = \frac{v_i(A_i)}{w_i}
    < \frac{v_i(A_j)}{w_j} \le \frac{\alpha_ip(A_j)}{w_j}$.
\item For all $g \in A_j$, $\displaystyle \frac{\alpha_i p(A_i)}{w_i} = \frac{v_i(A_i)}{w_i}
    < \frac{v_i(A_j \setminus \{g\})}{w_j} \le \frac{\alpha_ip(A_j \setminus \{g\})}{w_j}$.
\item For all $c \in A_i$, $\displaystyle
    \frac{\alpha_i p(A_i \setminus \{c\})}{w_i} = \frac{v_i(A_i \setminus \{c\})}{w_i}
    < \frac{v_i(A_j \setminus \{g\})}{w_j} \le \frac{\alpha_i p(A_j \setminus \{g\})}{w_j}$.
\end{enumerate}
Hence, agent $i$ pWEF1-envies agent $j$.
\end{proof}

\begin{corollary}
\label{thm:pef1-implies-ef1fpo}
Let $(A, p)$ be a market equilibrium.
If $(A, p)$ is pWEF1, then $A$ is WEF1+fPO.
\end{corollary}
\begin{proof}
Follows from \cref{thm:first-welfare,thm:ef1-envy-implies-pef1-envy}.
\end{proof}

One approach to finding a pWEF1 market equilibrium is to start with an integral market equilibrium
that may not be pWEF1 and iteratively modify it till it becomes pWEF1.
The modifications we consider are changing prices and transferring items from agent to another.
We formalize this idea using the notion of \emph{transferability} defined below.

\begin{definition}[transferability]
Let $(A, p)$ be an integral market equilibrium
for instance $\Ical \defeq ([n], [m], (v_i)_{i=1}^n, w)$.
Let $i \in [n]$ be an agent and $j \in [m] \setminus A_i$ be an item.
Let $B$ be an allocation where $B_i \defeq A_i \cup \{j\}$
and $B_{i'} \defeq A_{i'} \setminus \{j\}$ for all $i' \neq i$.
Then $j$ is said to be \emph{transferable} to agent $i$
if $(B, p)$ is also a market equilibrium.
\end{definition}

\begin{observation}
\label{thm:trn-check}
Let $(A, p)$ be an integral market equilibrium.
Let $i \in [n]$ be an agent and $j^* \in [m] \setminus A_i$ be a good or chore.
If $A_i$ contains a good or chore $j$, then
$j^*$ is transferable to $i$ iff $v_i(j^*)/p_{j^*} = v_i(j)/p_j$.
If $A_i$ contains only neutral items, then
$j^*$ is transferable to $i$ iff
\[ \frac{v_i(j^*)}{p_{j^*}} \in \left[ \max_{j: p_j > 0} \frac{v_i(j)}{p_j},
    \min_{j: p_j < 0} \frac{v_i(j)}{p_j} \right] \cap \mathbb{R}_{>0}. \]
\end{observation}

To find an integral market equilibrium to begin with,
we can just use the one that maximizes social welfare.

\begin{lemma}
\label{thm:soc-welf-ce}
Let $\Ical \defeq ([n], [m], (v_i)_{i=1}^n, w)$ be a fair division instance.
Let $A$ be the social-welfare maximizing allocation, i.e.,
each item $j$ is allocated to an agent in $\argmax_{i \in [n]} v_i(j)$.
For each $j \in A_i$, let $p_j \defeq v_i(j)$.
Then $(A, p)$ is a market equilibrium.
\end{lemma}
\begin{proof}
It is easy to check that if $j$ is a good, then $p_j > 0$,
if $j$ is neutral, then $p_j = 0$, and if $j$ is a chore, then $p_j < 0$.

Pick any agent $i \in [n]$. Let $\alpha_i = 1$.
For any item $j \in A_i$, we have $v_i(j) = p_j = \alpha_i p_j$.
For any item $j \in A_k$, for some $k \neq i$, we have $v_i(j) \le v_k(j) = \alpha_i p_j$.
Hence, $(A, p)$ is a market equilibrium.
\end{proof}

\subsection{Algorithm for Two Agents}
\label{sec:ef1-po:n2-loc-search}

We give an algorithm to compute a WEF1+fPO allocation of mixed manna
for two agents using Fisher markets.
See \cref{algo:n2-loc-search} for a precise description of the algorithm.

The algorithm starts with an arbitrary market equilibrium $(A, p)$.
First, it identifies the price-envious agent $\ell$. Let $b$ be the other agent.
Then it raises the prices of items in $A_{\ell}$,
while ensuring that $(A, p)$ is always a market equilibrium,
till either a good can be transferred from $b$ to $\ell$
or a chore can be transferred from $\ell$ to $b$.
Then it performs the transfer and repeats the process till $A$ becomes WEF1.

\begin{algorithm}[htb]
\caption{%
Takes as input a fair division instance
$\Ical \defeq ([n], [m], (v_i)_{i=1}^n, w)$ where $n=2$,
and an integral market equilibrium $(A, p)$.
Outputs a (supposedly) WEF1 market equilibrium.}
\label{algo:n2-loc-search}
\begin{algorithmic}[1]
\State Let $G$ be the set of goods and $C$ be the set of chores.
\While{$(A, p)$ is not WEF1} \label{alg-line:n2ls:while}
    \State $\ell = \argmin_{i \in [2]} p(A_i)/w_i$ and $b = 3 - \ell$.
    \If{$A_{\ell} \cap (G \cup C) \neq \emptyset$ and $A_b \cap (G \cup C) \neq \emptyset$}
        \LineComment{Raise prices of items in $A_{\ell}$ till a transfer becomes possible.}
        \State Let $\alpha_{\ell}$ and $\alpha_b$ be the
            best-bang-per-buck for $\ell$ and $b$, respectively.
        \State Let $\displaystyle \beta = \max_{g \in A_b \cap G} \frac{v_{\ell}(g)/p_g}{\alpha_{\ell}}$
            and $\displaystyle \gamma = \max_{c \in A_{\ell} \cap C} \frac{\alpha_b}{v_b(c)/p_c}$.
            \Comment{let $\max(\emptyset) \defeq -\infty$}
        \State Set $p_j = p_j / \max(\beta, \gamma)$ for all $j \in A_{\ell}$.
    \EndIf \label{alg-line:n2ls:pch}
    \If{$\exists \ghat \in A_b \cap G$ that is transferable to $\ell$ in $(A, p)$}
        \State \label{alg-line:n2ls:gt}Set $A_{\ell} = A_{\ell} \cup \{\ghat\}$ and $A_b = A_b \setminus \{\ghat\}$.
    \ElsIf{$\exists \chat \in A_{\ell} \cap C$ that is transferable to $b$ in $(A, p)$}
        \State \label{alg-line:n2ls:ct}Set $A_{\ell} = A_{\ell} \setminus \{\chat\}$ and $A_b = A_b \cup \{\chat\}$.
    \Else
        \State \label{alg-line:n2ls:err}\texttt{error}
    \EndIf
\EndWhile
\State \Return $(A, p)$.
\end{algorithmic}
\end{algorithm}

\begin{lemma}
\label{thm:n2ls:l-env-b}
At the beginning of an iteration of \cref{algo:n2-loc-search}, if $(A, p)$ is a market equilibrium,
then $\ell$ WEF1-envies $b$.
\end{lemma}
\begin{proof}
Based on the definition of $\ell$ and $b$, we get $p(A_b)/w_b \ge p(A_{\ell})/w_{\ell}$.
Hence, $b$ doesn't pWEF1-envy $\ell$ in $(A, p)$.
By \cref{thm:ef1-envy-implies-pef1-envy}, $b$ doesn't WEF1-envy $\ell$ in $A$.
Since $A$ is not WEF1 (line \ref{alg-line:n2ls:while}),
we get that $\ell$ WEF1-envies $b$.
\end{proof}

\begin{lemma}
\label{thm:n2ls:pch-meq}
If $(A, p)$ is a market equilibrium at the beginning of an iteration of \cref{algo:n2-loc-search},
then $(A, p)$ is a market equilibrium after line \ref{alg-line:n2ls:pch} in that iteration.
Moreover, $0 < \max(\beta, \gamma) \le 1$ in that iteration.
\end{lemma}
\begin{proof}
If $A_{\ell}$ or $A_b$ contains only neutral items,
then the prices are not modified, so $(A, p)$ remains a market equilibrium.
Now assume that both $A_{\ell}$ and $A_b$ have a non-neutral item.

If $v_{\ell}(j) \le 0$ for all $j \in A_b$ and $v_{\ell}(j) \ge 0$ for all $j \in A_{\ell}$, then
\[ \frac{v_{\ell}(A_{\ell})}{w_{\ell}} \ge 0 \ge \frac{v_{\ell}(A_b)}{w_b}, \]
which contradicts \cref{thm:n2ls:l-env-b}.
Hence, either $v_{\ell}(g) > 0$ for some $g \in A_b$
or $v_{\ell}(c) < 0$ for some $c \in A_{\ell}$ (or both).
Hence, $\beta > 0$ or $\gamma > 0$.

Since $(A, p)$ is a market equilibrium, we get
$v_{\ell}(g) \le \alpha_{\ell}p_g$ for all $g \in A_b \cap G$,
and $(-v_b(c)) \ge \alpha_b(-p_c)$ for all $c \in A_{\ell} \cap C$.
Hence, $\beta \le 1$ and $\gamma \le 1$.

Let $\rho \defeq \max(\beta, \gamma)$. Then $\rho \in (0, 1]$.
Let $\phat_j \defeq p_j/\rho$ for $j \in A_{\ell}$ and $\phat_j \defeq p_j$ for $j \in A_b$.
Let $\alphahat_{\ell} \defeq \alpha_{\ell}\rho$ and $\alphahat_b \defeq \alpha_b$.
We will show that $(A, \phat)$ is also a market equilibrium,
with $\alphahat_{\ell}$ and $\alphahat_b$ being the corresponding
best-bang-per-buck for agents $\ell$ and $b$, respectively.

For $g \in A_{\ell} \cap G$, we get
\begin{align*}
\frac{v_{\ell}(g)}{\phat_g} &= \frac{v_{\ell}(g)}{p_g}\rho = \alphahat_{\ell},
& \frac{v_b(g)}{\phat_g} &= \frac{v_b(g)}{p_g}\rho \le \alpha_b\rho \le \alphahat_b.
\end{align*}
For $c \in A_{\ell} \cap C$, we get
\begin{align*}
\frac{v_{\ell}(c)}{\phat_c} &= \frac{v_{\ell}(c)}{p_c}\rho = \alphahat_{\ell},
& \frac{v_b(c)}{\phat_c} &= \frac{v_b(c)}{p_c}\rho \ge \frac{\alpha_b}{\gamma}\rho \ge \alphahat_b.
\end{align*}
For $g \in A_b \cap G$, we get
\begin{align*}
\frac{v_{\ell}(g)}{\phat_g} &\le \alpha_{\ell}\beta \le \alphahat_{\ell},
& \frac{v_b(g)}{\phat_g} &= \alpha_b = \alphahat_b.
\end{align*}
For $c \in A_b \cap C$, we get
\begin{align*}
\frac{v_{\ell}(c)}{\phat_c} &\ge \alpha_{\ell} \ge \alphahat_{\ell},
& \frac{v_b(c)}{\phat_c} &= \alpha_b = \alphahat_b.
\end{align*}
Hence, $(A, \phat)$ is a market equilibrium.
\end{proof}

\begin{lemma}
\label{thm:n2ls:meq}
$(A, p)$ is a market equilibrium throughout \cref{algo:n2-loc-search}.
\end{lemma}
\begin{proof}
Suppose $(A, p)$ is a market equilibrium at the beginning of an iteration.
(This is true for the first iteration, by assumption.)
It remains a market equilibrium after line \ref{alg-line:n2ls:pch} by \cref{thm:n2ls:pch-meq}.
It remains a market equilibrium for the rest of the iteration by the definition of transferability.
Hence, each iteration preserves market equilibrium.
\end{proof}

\begin{lemma}
\label{thm:n2ls:no-err}
\Cref{algo:n2-loc-search} never executes line \ref{alg-line:n2ls:err}.
\end{lemma}
\begin{proof}
Let $(A, p)$ be the market equilibrium at the beginning of an iteration
and let $\phat$ be the modified price vector after line \ref{alg-line:n2ls:pch}.

Suppose $A_{\ell}$ contains only neutral items.
If $v_{\ell}(j) \le 0$ for all $j \in A_b$, then
$v_{\ell}(A_{\ell})/w_{\ell} \ge 0 \ge v_{\ell}(A_b)/w_b$,
which contradicts \cref{thm:n2ls:l-env-b}.
Hence, $v_{\ell}(g) > 0$ for some $g \in A_b$.
By \cref{thm:trn-check}, some $\ghat \in A_b$ is transferable to $\ell$ in $(A, \phat)$,
so line \ref{alg-line:n2ls:gt} is executed.

Suppose $A_b$ contains only neutral items.
If $p_j \ge 0$ for all $j \in A_{\ell}$, then
$p(A_{\ell})/w_{\ell} \ge 0 \ge p(A_b)/w_b$,
which contradicts the definition of $\ell$ and $b$.
Hence, $A_{\ell} \cap C \neq \emptyset$.
By \cref{thm:trn-check}, some $\chat \in A_{\ell}$ is transferable to $b$ in $(A, \phat)$,
so line \ref{alg-line:n2ls:ct} is executed.

Suppose both $A_{\ell}$ and $A_b$ contain non-neutral items.
By \cref{thm:n2ls:pch-meq}, $\rho \in (0, 1]$, where $\rho \defeq \max(\beta, \gamma)$.

\textbf{Case 1}: $\rho = \beta$.
\\ Then $A_b \cap G \neq \emptyset$.
Let $\ghat \in \argmax_{g \in A_b \cap G} v_{\ell}(g)/(\alpha_{\ell}p_g)$.
Then $\phat_{\ghat} = p_{\ghat}$, so
$v_{\ell}(\ghat)/\phat_{\ghat} = \alpha_{\ell}\beta$,
and for all $j \in A_{\ell} \cap (G \cup C)$, we have
$v_{\ell}(j)/\phat_j = \rho v_{\ell}(j)/p_j = \alpha_{\ell}\beta$.
Hence, by \cref{thm:trn-check}, $\ghat$ is transferable to $\ell$ in $(A, \phat)$,
so line \ref{alg-line:n2ls:gt} is executed.

\textbf{Case 2}: $\rho = \gamma$.
\\ Then $A_{\ell} \cap C \neq \emptyset$.
Let $\chat \in \argmax_{c \in A_{\ell} \cap C} \alpha_bp_c/v_{\ell}(c)$.
Then $\phat_{\chat} = \rho p_{\chat}$, so
$v_b(\chat)/\phat_{\chat} = \rho\alpha_b/\gamma = \alpha_b$,
and for all $j \in A_b \cap (G \cup C)$, we have
$v_b(j)/\phat_j = v_b(j)/p_j = \alpha_b$.
Hence, by \cref{thm:trn-check}, $\chat$ is transferable to $b$ in $(A, \phat)$,
so line \ref{alg-line:n2ls:ct} is executed.

In all cases we discussed above, either line \ref{alg-line:n2ls:gt}
or line \ref{alg-line:n2ls:ct} is executed.
Hence, line \ref{alg-line:n2ls:err} is not executed.
\end{proof}

\begin{lemma}
\label{thm:n2ls:same-lb}
The pair $(\ell, b)$ is the same for all iterations of \cref{algo:n2-loc-search}.
\end{lemma}
\begin{proof}
Assume \wLoG{} that $\ell = 1$ and $b = 2$ at the beginning of
an iteration of \cref{algo:n2-loc-search}.
Let $(A^{(1)}, p^{(1)})$ be the value of $(A, p)$ at the beginning of the iteration,
and $(A^{(2)}, p^{(2)})$ be the value of $(A, p)$ at the end of the iteration.
Then $(A^{(1)}, p^{(2)})$ is the value of $(A, p)$ after line \ref{alg-line:n2ls:pch}.
They are all market equilibria by \cref{thm:n2ls:meq}.

By \cref{thm:n2ls:l-env-b}, 1 WEF1-envies 2 in $A^{(1)}$.
By \cref{thm:pef1-implies-ef1fpo}, 1 pWEF1-envies 2 in
$(A^{(1)}, p^{(1)})$ and $(A^{(1)}, p^{(2)})$.

By \cref{thm:n2ls:no-err}, either line \ref{alg-line:n2ls:gt} or
line \ref{alg-line:n2ls:ct} is executed.
If line \ref{alg-line:n2ls:gt} is executed,
then since 1 pWEF1-envies 2 in $(A^{(1)}, p^{(2)})$, we get
\[ \frac{p^{(2)}(A^{(2)}_1 \setminus \{\ghat\})}{w_1}
    = \frac{p^{(2)}(A^{(1)}_1)}{w_1}
    < \frac{p^{(2)}(A^{(1)}_b \setminus \{\ghat\})}{w_b}
    = \frac{p^{(2)}(A^{(2)}_b)}{w_b}. \]
If line \ref{alg-line:n2ls:ct} is executed,
then since 1 pWEF1-envies 2 in $(A^{(1)}, p^{(2)})$, we get
\[ \frac{p^{(2)}(A^{(2)}_1)}{w_1}
    = \frac{p^{(2)}(A^{(1)}_1 \setminus \{\chat\})}{w_1}
    < \frac{p^{(2)}(A^{(1)}_b)}{w_b}
    = \frac{p^{(2)}(A^{(2)}_b \setminus \{\chat\})}{w_b}. \]
Hence, 2 doesn't pWEF1-envy 1 in $(A^{(2)}, p^{(2)})$.

By \cref{thm:ef1-envy-implies-pef1-envy}, 2 doesn't WEF1-envy 1 in $A^{(2)}$.
In the next iteration (if it takes place), agent $\ell$ WEF1-envies agent $b$,
by \cref{thm:n2ls:l-env-b}. Hence, $\ell$ must be 1 and $b$ must be 2.
\end{proof}

\begin{lemma}
\label{thm:n2ls:iters}
\Cref{algo:n2-loc-search} terminates after at most $m$ iterations.
\end{lemma}
\begin{proof}
Assume \wLoG{} that $\ell = 1$ and $b = 2$ in each iteration
of \cref{algo:n2-loc-search} (by \cref{thm:n2ls:same-lb}).
In each iteration, either a good is transferred from agent 2 to agent 1,
or a chore is transferred from agent 1 to agent 2,
so $|A_1 \cap C| + |A_2 \cap G|$ decreases in each iteration.
Its value is at most $m$ initially,
so the algorithm terminates after at most $m$ iterations.
\end{proof}

\begin{theorem}
\Cref{algo:n2-loc-search} outputs a WEF1+fPO allocation in polynomial time.
\end{theorem}
\begin{proof}
By \cref{thm:n2ls:no-err}, \cref{algo:n2-loc-search} doesn't throw an error.
Hence, by \cref{thm:n2ls:meq}, \cref{algo:n2-loc-search} outputs a market equilibrium $(A, p)$.
By \cref{thm:first-welfare}, $A$ is fPO.
The algorithm breaks out of the while loop only when the allocation $A$ is WEF1.
Hence, $A$ is WEF1+fPO.

By \cref{thm:n2ls:iters}, the algorithm runs for at most $m$ iterations.
Each iteration can be executed in polynomial time, since
transferability can be checked efficiently using \cref{thm:trn-check}.
Hence, \cref{algo:n2-loc-search} runs in polynomial time.
\end{proof}

\section{Conclusion and Open Problems}
\label{sec:conclusion}

We saw that obtaining an EF1 allocation becomes significantly more challenging when
agents have unequal entitlements and the items are mixed manna.
However, we show that a slightly weaker notion of fairness,
called \emph{Weighted Envy-free up to One Transfer} (WEF1T), can still be achieved.
Whether a WEF1 allocation of mixed manna always exists,
even for three agents, is still an open problem.

For the special case of two agents, we obtain not just WEF1,
but also fractional Pareto optimality.
We do this by extending the concept of market equilibrium to the mixed manna setting.
It is not yet known whether an EF1+PO allocation of mixed manna always exists for three agents,
even for equal entitlements.

\end{document}